\documentclass[11pt]{article}

\usepackage{amssymb,amsthm,amsmath}
\usepackage{graphicx}
\usepackage{url}

\usepackage[small]{caption}
\usepackage{subcaption}

\usepackage[utf8]{inputenc}
\usepackage{amsfonts}
\usepackage{fullpage}
\usepackage{framed}

\newtheorem{theorem}{Theorem}
\newtheorem{lemma}{Lemma}
\newtheorem{proposition}{Proposition}

\newcommand{\eps}{\varepsilon}

\newcommand{\NN}{\mathbb{N}} 

\def\A{\mathcal{A}}

\def\idle{{\tt idle}}

\newcommand{\etal}{{et~al.}}
\newcommand{\ie}{{i.e.}}
\newcommand{\eg}{{e.g.}}

\newcommand{\old}[1]{}

\title{On Fence Patrolling by Mobile Agents\thanks{A preliminary
    version of this paper appeared in the {\em Proceedings of the 25th
      Canadian Conference on Computational Geometry}, Waterloo, ON,
    Canada, August 2013.
Research supported in part by the NSF grant DMS-1001667.}}

\author{%
Adrian Dumitrescu\thanks{%
Department of Computer Science,
University of Wisconsin--Milwaukee, USA\@.
Email:~\texttt{dumitres@uwm.edu}.}
\and
Anirban Ghosh\thanks{%
Department of Computer Science,
University of Wisconsin--Milwaukee, USA\@.
Email:~\texttt{anirban@uwm.edu}.}
\and
Csaba D. T\'oth\thanks{%
Department of Mathematics,
California State University, Northridge, Los Angeles, USA\@.
Email:  \texttt{cdtoth@acm.org}.}
}

\begin{document}

\maketitle

\begin{abstract}
Suppose that a fence needs to be protected (perpetually) by $k$ mobile
agents with maximum speeds $v_1,\ldots,v_k$ so that no point on the
fence is left unattended for more than a given amount of time. The
problem is to determine if this requirement can be met, and if so, to design a
suitable patrolling schedule for the agents.
Alternatively, one would like to find a schedule that minimizes the
\emph{idle time}, that is, the longest time interval during which
some point is not visited by any agent. We revisit this problem,
introduced by Czyzowicz et al.~(2011), and discuss several strategies
for the cases where the fence is an open and a closed curve, respectively.

In particular:
(i) we disprove a conjecture by Czyzowicz et al. regarding the
optimality of their Algorithm~$\A_2$ for unidirectional patrolling of
a closed fence;
(ii) we present an algorithm with a lower idle time for patrolling
an open fence, improving an earlier result of Kawamura and Kobayashi.


\end{abstract}

\section{Introduction} \label{sec:intro}

A set of $k$ mobile agents with (possibly distinct) maximum speeds
$v_i$ ($i=1,\ldots,k$)  are in charge of \emph{guarding} or in
other words \emph{patrolling} a given region of interest.
Patrolling problems find applications in the field of robotics where
surveillance of a region is necessary.
An interesting one-dimensional variant have been introduced by
Czyzowicz~\etal~\cite{CGKK11}, where the agents move along a
rectifiable Jordan curve representing a \emph{fence}. The fence is
either a \emph{closed} curve (the boundary of a compact region in the
plane), or an \emph{open} curve (the boundary between two regions). For simplicity
(and without loss of generality) it can be assumed that the open curve is a
line segment and the closed curve is a circle. The movement of the agents
over the time interval $[0,\infty)$ is described by a \emph{patrolling schedule},
where the speed of the $i$th agent, $a_i$ ($i=1,\ldots,k$), may vary between
zero and its maximum value $v_i$ in any of the two moving directions
(left or right).

Given a closed or open fence of length $\ell$ and maximum speeds
$v_1,\ldots,v_k>0$ of $k$ agents, the goal is to find a
\emph{patrolling schedule} that minimizes the \emph{idle time} $I$,
defined as the longest time interval in $[0,\infty)$ during which a
  point on the fence remains unvisited, taken over all points. A
  straightforward volume argument~\cite{CGKK11} yields the lower bound
  $I \geq \ell/\sum_{i=1}^{k} v_i$ for an (open or closed) fence of length~$\ell$.
A \emph{patrolling algorithm} computes a \emph{patrolling schedule} 
for a given fence and set of speeds $v_1,\ldots,v_k>0$.

For an open fence (line segment), Czyzowicz~\etal~\cite{CGKK11}
proposed a simple partitioning strategy, algorithm $\A_1$,
where each agent moves back and forth perpetually in a segment whose length
is proportional with its speed. Specifically, for a segment of
length $\ell$ and $k$ agents with maximum speeds $v_1,\ldots ,v_k$,
algorithm $\A_1$ partitions the segment into $k$ pieces of lengths
$\ell v_i/\sum_{j=1}^kv_j$, and schedules the $i$th agent to patrol
the $i$th interval with speed $v_i$.
Algorithm $\A_1$ has been proved to be optimal for uniform
speeds~\cite{CGKK11}, \ie, when all maximum speeds are equal.
Algorithm $\A_1$ achieves an idle time $2\ell/\sum_{i=1}^{k} v_i$
on a segment of length $\ell$, and so $\A_1$ is a 2-approximation
algorithm for the shortest idle time.
It has been conjectured~\cite[Conjecture~1]{CGKK11} that $\A_1$ is optimal
for arbitrary speeds, however this was disproved by Kawamura and
Kobayashi~\cite{KK12}: they selected speeds $v_1,\ldots,v_6$
and constructed a schedule for $6$ agents that achieves an idle time of
$\frac{41}{42} \, \left(2\ell/\sum_{i=1}^{k} v_i \right)$.

A patrolling algorithm $\A$ is \emph{universal} if it can be executed
with any number of agents $k$ and any speed setting $v_1,\ldots,v_k>0$
for the agents. For example, $\A_1$ above is universal, however certain algorithms
(\eg, algorithm $\A_3$ in Section~\ref{sec:2/3} or the algorithm in
Section~\ref{sec:24/25}) can only be executed with certain speed
settings or number of agents, \ie, they are not universal.

For the closed fence (circle), no universal algorithm has been
proposed to be optimal. For uniform speeds (\ie, $v_1=\ldots=v_k=v$), 
it is not difficult to see that placing the agents uniformly around the
circle and letting them move in the same direction yields the shortest idle time.
Indeed, the idle time in this case is $\ell/(kv)= \ell/\sum_{i=1}^{k} v_i$,
matching the lower bound mentioned earlier.

For the variant in which all agents are \emph{required} to move in the
same direction along a circle of unit length (say clockwise),
Czyzowicz~\etal~\cite[Conjecture~2]{CGKK11} conjectured that the
following algorithm $\A_2$ always yields an optimal schedule.
Label the agents so that $v_1 \geq v_2 \geq \ldots \geq v_k>0$.
Let $r$, $1\leq r\leq k$, be an index such that $\max_{1 \leq i \leq k} i v_i =r v_r$.
Place the agents at equal distances of $1/r$ around the circle,
so that each moves clockwise at the same speed $v_r$. Discard the
remaining agents, if any. Since all agents move in the same direction,
we also refer to $\A_2$ as the ``runners'' algorithm. It achieves an
idle time of $1/\max_{1 \leq i \leq k} i v_i$~\cite[Theorem~2]{CGKK11}.
Observe that $\A_2$ is also universal. 

\paragraph{Historical perspective.} 
Multi-agent patrolling is a variation of the fundamental problem 
of multi-robot coverage~\cite{Ch04,Ch01}, studied extensively in the 
robotics community. A variety of models has been considered for 
patrolling, including deterministic and randomized, as well as  
centralized and distributed strategies, under various 
objectives~\cite{AKK08,EAK07}. \emph{Idleness}, as a measure of 
efficiency for a patrolling strategy, was introduced by 
Machado~\etal~\cite{MRZD02} in a graph setting; see also the article
by Chevaleyre~\cite{Ch04}.

The closed fence patrolling problem is reminiscent of the classical 
\emph{lonely runners conjecture}, introduced by Wills~\cite{Wi67} and
Cusick~\cite{Cu73}, independently, in number theory and discrete
geometry. Assume that $k$ agents run clockwise along a circle of
length $1$, starting from the same point at time $t=0$. 
They have distinct but constant speeds (the speeds cannot vary, unlike in the 
model considered in this paper). A runner is called \emph{lonely} when he/she is 
at distance of at least $\frac{1}{k}$ from any other runner (along the circle). 
The conjecture asserts that each runner $a_i$ is lonely at some time
$t_i\in (0,\infty)$. The conjecture has only been confirmed for up to
$k=7$ runners~\cite{BS08,BHK01}. 

\paragraph{Notation and terminology.}
A \emph{unit} circle is a circle of unit length.
We parameterize a line segment and a circle of length $\ell$
by the interval $[0,\ell]$.
A \emph{schedule} of $k$ agents consists of $k$ functions
$f_i:[0,\infty] \rightarrow [0,\ell]$, for 
$i=1,\ldots , k$, where $f_i(t)$ is the position of agent
$i$ at time $t$. Each function $f_i$ is continuous
(for a closed fence, the endpoints of the interval $[0,\ell]$ are
identified), it is piecewise differentiable, and its derivative
(speed) is bounded by $|f_i'|\leq v_i$.
A schedule is called \emph{periodic} with \emph{period} $t_0>0$
if $f_i(t)=f_i(t+t_0)$ for all $i=1,\ldots , k$ and
$t\geq 0$. The \emph{idle time} $I$ of a schedule
is the length of the maximum (open) time interval
$(t_1,t_2)$ such that there is a point $x\in [0,\ell]$
where $f_i(t)\neq x$ for all $i=1,\ldots , k$ and
$t\in (t_1,t_2)$. For a given fence (closed or open)
of length $\ell$ and given maximum speeds $v_1,\ldots,v_k$,
$\idle(\A)$ denotes the idle time of a schedule produced by
algorithm $\A$.

We use \emph{position-time diagrams} to plot the agent trajectories
with respect to time. One axis represents the position $f_i(t)$ of the
agents along the fence and the other axis represents time.
In Fig.~\ref{fig:dtd}, for instance, the horizontal axis represents the position
of the agents along the fence and the vertical axis represents time.
In Fig.~\ref{fig:circle}, however, the vertical axis represents the position
and the vertical axis represents time.
A schedule with idle time $I$ is equivalent to a covering problem
in such a diagram (see Fig.~\ref{fig:dtd}). For a straight-line
(\ie, constant speed) trajectory between points $(x_1,y_1)$ and $(x_2,y_2)$
in the diagram, construct a shaded parallelogram with vertices, $(x_1,y_1)$,
$(x_1,y_1+I)$, $(x_2,y_2)$, $(x_2,y_2+I)$, where $I$ denotes the desired idle
time and the shaded region represents the covered region. In particular,
if an agent stays put in a time-interval, the parallelogram degenerates
to a vertical segment. A schedule for the agents ensures idle time $I$
if and only if the entire area of the diagram in the time interval $[I,\infty)$
is covered.

To evaluate the efficiency of a patrolling algorithm $\A$, we use the ratio
$\rho=\idle(\A)/\idle(\A_1)$ between the idle times of $\A$ and
the partition-based algorithm $\A_1$. Lower values of $\rho$ indicate
better (more efficient) algorithms. Recall however that certain
algorithms can only be executed with certain speed settings or number
of agents.

We write $H_n= \sum_{i=1}^n 1/i$ for the $n$th \emph{harmonic number}.

\begin{figure}[htbp]
    \begin{center}
          \includegraphics[width=0.25\textwidth]{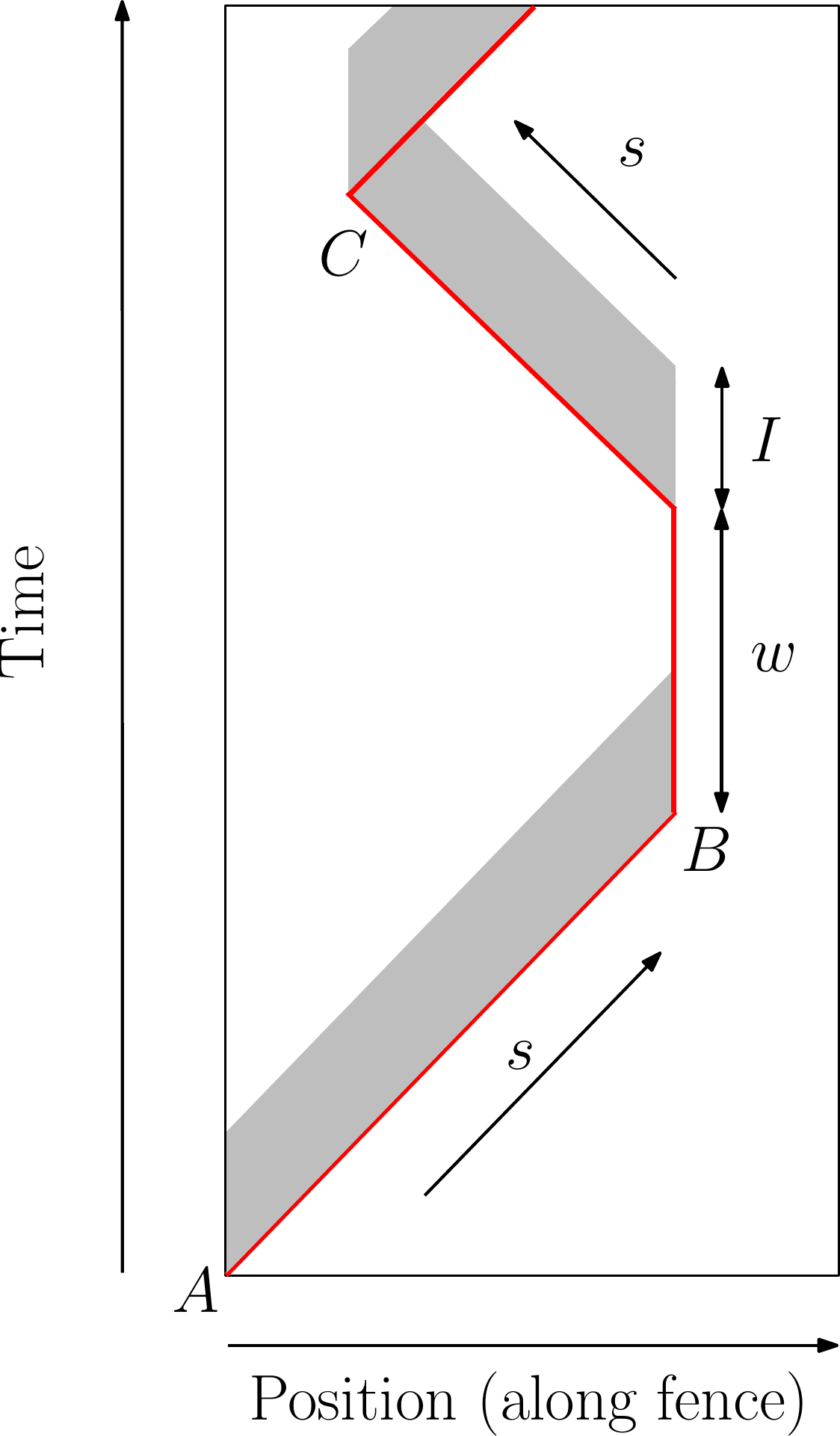}
    \end{center}
\caption{Agent moving with speed $s$ from $A$ to
$B$, waiting at $B$ for time $w$ and then
moving from $B$ to $C$ with speed $s$.}
\label{fig:dtd}
\end{figure}
%


\paragraph{Our results.}
\begin{enumerate}

\item Consider the unidirectional unit circle
(where all agents are required to move in the same direction).

(i) We disprove a conjecture by Czyzowicz~\etal~\cite[Conjecture~2]{CGKK11}
regarding the optimality of Algorithm~$\A_2$. Specifically, we construct
a schedule for $32$ agents with harmonic speeds $v_i=1/i$, $i=1,\ldots,32$,
that has an idle time strictly less than $1$.
In contrast, Algorithm~$\A_2$ yields a unit idle time for harmonic speeds
($\idle(\A_2)=1$), hence it is suboptimal. See Theorem~\ref{thm:counter},
Section~\ref{sec:uni}.

(ii) For every $\tau\in (0,1]$ and $t \geq \tau$, there exists
a positive integer $k=k(t) \leq e^{4t/\tau^2}$ and a schedule for the system
of $k$ agents with harmonic speeds $v_i=1/i$, $i=1,\ldots,k$, that ensures
an idle time at most $\tau$ during the time interval $[0,t]$.
See Theorem~\ref{thm:tau}, Section~\ref{sec:uni}.

\item Consider the open fence patrolling. For every integer $x \geq 2$,
there exist $k=4x+1$ agents with $\sum_{i=1}^{k} v_i=16x+1$ and a guarding
schedule for a segment of length $25x/3$. Alternatively, for every integer
$x \geq 2$ there exist $k=4x+1$ agents with suitable speeds $v_1,\ldots,v_k$,
and a guarding schedule for a unit segment that achieves idle
time at most $\frac{48x+3}{50x} \, \frac{2}{\sum_{i=1}^{k} v_i}$.
In particular, for every $\eps>0$, there exist $k$ agents
with suitable speeds $v_1,\ldots,v_k$, and a guarding schedule for
a unit segment that achieves idle time at most
$\left(\frac{24}{25} +\eps\right) \, \frac{2}{\sum_{i=1}^{k} v_i}$.
This improves the previous bound of $\frac{41}{42} \, \frac{2}{\sum_{i=1}^{k} v_i}$
by Kawamura and Kobayashi~\cite{KK12}. See Theorem~\ref{thm:24/25},
Section~\ref{sec:24/25}.

\item Consider the bidirectional unit circle.

(i) For every $k \geq 4$, there exist maximum speeds
$v_1 \geq v_2 \geq \ldots \geq v_k$ and a new patrolling algorithm $\A_3$
that yields an idle time better than that achieved by both $\A_1$ and $\A_2$.
In particular, for large $k$, the idle time of $\A_3$ with these
speeds is about $2/3$ of that achieved by $\A_1$ and $\A_2$.
See Proposition~\ref{prop:2/3}, Section~\ref{sec:2/3}.

(ii) For every $k \geq 2$, there exist maximum speeds
$v_1 \geq v_2 \geq \ldots \geq v_k$ so that there exists an optimal
  schedule (patrolling algorithm) for the circle that does not use
up to $k-1$ of the agents $a_2,\ldots,a_k$.
In contrast, for a segment, any optimal schedule must use all agents.
See Proposition~\ref{prop:useless}, Section~\ref{sec:2/3}.

(iii) There exist settings in which if all $k$ agents are used
by a patrolling algorithm, then some agent(s) need overtake (pass)
other agent(s). This partially answers a question left open by
Czyzowicz~\etal~\cite[Section~3]{CGKK11}.
See the remark at the end of Section~\ref{sec:2/3}.

\end{enumerate}

\section{Unidirectional Circle Patrolling} \label{sec:uni}

\paragraph{A counterexample for the optimality of algorithm $\A_2$.}
We show that Algorithm~$\A_2$ by Czyzowicz~\etal~\cite{CGKK11}
for unidirectional circle patrolling is not always optimal.
We consider agents with \emph{harmonic speeds} $v_i=1/i$, $i\in \NN$.
Obviously, for this setting we have $\idle(\A_2) =1$, which is
already achieved by the agent $a_1$ with the highest (here unit) speed.
We design a periodic schedule (patrolling algorithm) for $k=32$ agents
with idle time $I<1$. In this schedule, agent $a_1$ moves continuously
with unit speed, and it remains to schedule agents $a_2,\ldots,a_{32}$ such
that every point is visited at least one more time in the unit length
open time interval between two consecutive visits of $a_1$. We start
with a weaker claim, for \emph{closed} intervals but using only $6$ agents.

\begin{lemma}\label{lem:counter}
Consider the unit circle, where all agents are required to move in the
same direction. For k=6 agents of harmonic speeds $v_i=1/i$, $i=1,\ldots ,6$,
there is a schedule where agent $a_1$ moves continuously with speed $1$, and
every point on the circle is visited by some other agent in every
closed unit length time interval between two consecutive visits of $a_1$.
\end{lemma}
\begin{figure}[htbp]
    \begin{center}
          \includegraphics[width=0.98\textwidth]{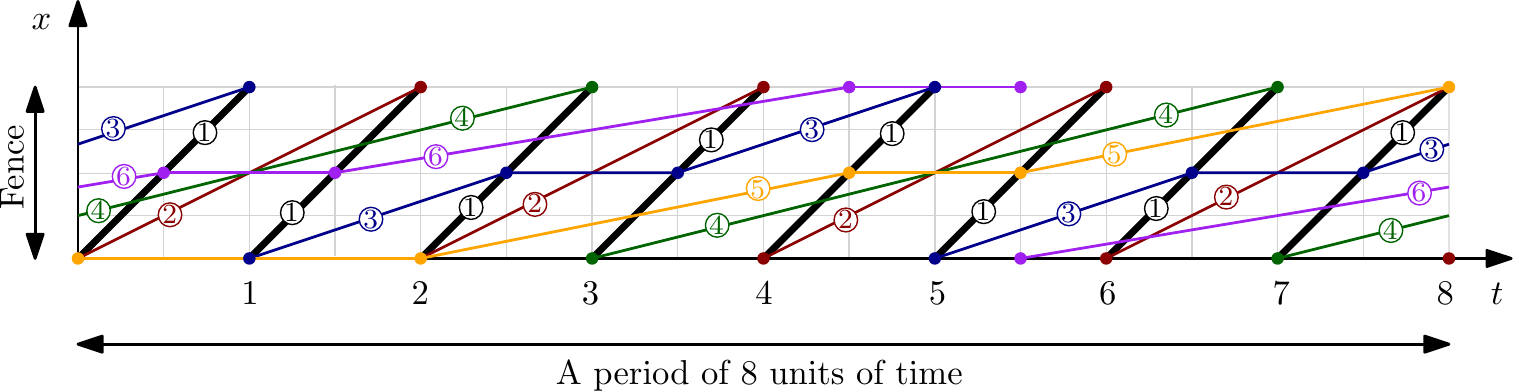}
    \end{center}
\caption{A periodic schedule of $6$ agents of speeds $1/i$, $i=1,\ldots,6$,
on a unit circle with period $8$. Agent $a_1$ moves continuously with speed $1$.
Each point is visited by one of the agents $a_2,a_3,a_4,a_5,a_6$ between any
two consecutive visits of agent $a_1$.}
\label{fig:circle}
\end{figure}

\begin{proof}
Our proof is constructive. We construct a periodic schedule for the $6$
agents with period $8$; refer to Fig.~\ref{fig:circle}. Agents $a_1$,
$a_2$ and $a_4$ continuously move with maximum speed,
while agents $a_3$, $a_5$ and $a_6$ each stop at certain times in
their movements.
Their schedule in one period $t\in [0,8]$ is given by the following
piecewise linear functions.
$$f_1(t)= t\mod 1, \hspace{1cm} f_2(t)=t/2\mod 1, \hspace{1cm}
f_4(t)= (t-3)/4 \mod 1.$$
$$f_3(t)=
\left\{ \begin{array}{ll}
(t-1)/3 \mod 1 & \mbox{\rm for } t\in [0,2.5]\cup [7.5,8]\\
0.5           &  \mbox{\rm for } t\in [2.5,3.5]\cup [6.5,7.5]\\
(t-2)/3 \mod 1 & \mbox{\rm for } t\in [3.5,6.5]\\
t/3 \mod 1 & \mbox{\rm for } t\in [7.5,8].
\end{array}\right.$$
$$f_5(t)=
\left\{ \begin{array}{ll}
0              & \mbox{\rm for } t\in [0,2]\\
(t-2)/5 \mod 1 & \mbox{\rm for } t\in [2,4.5]\\
0.5            &  \mbox{\rm for } t\in [4.5,5.5]\\
(t-3)/5 \mod 1 &  \mbox{\rm for } t\in [5.5,8].
\end{array}\right.$$
$$f_6(t)=
\left\{ \begin{array}{ll}
(t-3.5)/6 \mod 1 & \mbox{\rm for } t\in [0,0.5]\\
0.5            &  \mbox{\rm for } t\in [0.5,1.5]\\
(t-4.5)/6 \mod 1 &  \mbox{\rm for } t\in [1.5,4.5]\\
1            &  \mbox{\rm for } t\in [4.5,5.5]\\
(t-5.5)/6 \mod 1 &  \mbox{\rm for } t\in [5.5,8].
\end{array}\right.$$
\end{proof}

\begin{theorem}\label{thm:counter}
Consider the unit circle, where all agents are required to move in the
same direction. For $32$ agents of harmonic speeds $v_i=1/i$, $i=1,\ldots ,32$,
there is a periodic schedule with idle time strictly less than $1$.
\end{theorem}
\begin{proof}
Agents $a_1,\ldots,a_6$ follow the periodic schedule described in
Lemma~\ref{lem:counter}.
A time-position pair $(t,x)\in [0,8)\times [0,1)$ is a \emph{critical point}
in the time-position diagram if point $x$ on the fence is \emph{not}
traversed by any agent in the open time interval $(t,t+1)$.
There are exactly $12$ critical points in the schedule in Fig.~\ref{fig:circle}.
Specifically, these points are
$(j,0)$ for $j=0,1,\ldots, 7$; and $(j+\frac{1}{2},\frac{1}{2})$ for
$j=1,3,5,7$.

For each critical point $(t,x)$, we assign one, two, or four agents
such that they jointly traverse a small neighborhood of the critical
point in each period in the periodic schedule.

We schedule agents $a_7$ and $a_8$ to move continuously with speed $1/8$, as follows.
$$f_7(t)= \frac{1}{8}\left(t-\frac{1}{3}\right)\mod 1,
\hspace{1cm}
f_8(t)=\frac{1}{8}\left(t-\frac{7}{3}\right)\mod 1.$$
Agent $a_7$ traverses the unit intervals of the critical points
$(0,0)$ and $(3+\frac{1}{2},\frac{1}{2})$; and agent $a_8$ traverses
the unit intervals of the critical points $(2,0)$ and
$(5+\frac{1}{2},\frac{1}{2})$.
We are left with $8$ critical points, which will be taken care of by
agents $a_9,\ldots , a_{32}$.

Agents $a_9,\ldots,a_{16}$ are scheduled to move with constant speed $1/16$.
These $8$ agents form $4$ pairs, where each pair is responsible to
visit the neighborhood of a critical point in each period of length $8$
(each agent in a pair returns to the same critical point
after $16$ units of time). Finally, agents $a_{17},\ldots,a_{32}$ move
with constant speed $1/32$. These $16$ agents form 4 quadruples, where
each quadruple is responsible to visit the neighborhood of a critical
point in each period of $8$ (each agent in a quadruple returns after $32$ units of
time).

This schedule ensures that every point on the fence within a small
neighborhood of the $12$ critical points is visited by some agent within every time
interval of length $1-\eps$, where $\eps>0$ is a sufficiently small
constant. Apart from these neighborhoods, the first $6$ agents already
visit every point within every time interval of length $1-\eps$ if
$\eps>0$ is sufficiently small.
\end{proof}

\paragraph{Remark.} In Theorem~\ref{thm:counter}, we required that
all agents move in the same direction (clockwise) along the unit circle,
but we allowed agents to stop (\ie, have zero speed). If all agents
are required to maintain a \emph{strictly} positive speed, the proof
of Theorem~\ref{thm:counter} would still go through: in this case,
agents $a_3$, $a_5$ and $a_6$ could move at an extremely slow but
positive speed instead of stopping. As a result, some points at the
neighborhoods of the $12$ critical points would remain unvisited for
1 unit of time (this frequency is maintained by agent $a_1$ alone).
However, agents $a_7,\ldots , a_{32}$ would still ensure that every
point in these neighborhoods is also visited within every time
interval of length $1-\eps$.

\paragraph{Finite time patrolling.}
Interestingly enough, we can achieve any prescribed idle time below $1$
for an arbitrarily long time in this setting, provided we choose the
number of agents $k$ large enough.

\begin{theorem} \label{thm:tau}
Consider the unit circle, where all agents are required to move in the
same direction. For every $0<\tau \leq 1$ and $t \geq \tau$, there exists
$k=k(t) \leq e^{4t/\tau^2}$ and a schedule for the system
of $k$ agents with maximum speeds $v_i=1/i$, $i=1,\ldots,k$, that ensures
an idle time $\leq \tau$ during the time interval $[0,t]$.
\end{theorem}
\begin{proof}
We construct a schedule with an idle time at most $\tau$.
Let agent $a_1$ start at time $0$ and move clockwise at maximum (unit)
speed, \ie, $a_1(t)=t \mod 1$ denotes the position on the unit circle
of agent $a_1$ at time $t$.
Assume without loss of generality that $t$ is a multiple
of $\tau$, \ie, $t=m \tau$, where $m$ is a natural number.
Divide the time interval $[0,t]$ into $2m$ subintervals of length $\tau/2$.
For $j=1,\ldots,2m$, $[(j-1)\tau/2,j\tau/2]$ is the $j$th interval.

For each $j$, cover the unit circle $C$ so that every point
of $C$ is visited at least once by some agent. This ensures that
each point of the circle is visited at least once in the time interval
$[0,\tau/2]$ and no two consecutive visits to any one point are
separated in time by more than $\tau$ thereafter until time $t$,
as required.

To achieve the covering condition in each interval $j$, we use the first
agent ($a_1$, of unit speed), and as many other unused agents as needed.
The `origin' on  $C$ is reset to the current position of $a_1$ at
time $(j-1)\tau/2$, \ie, the beginning of the current time interval.
So the fastest agent is used (continuously) in all $2m$ time intervals.
Agent $a_1$ can cover a distance of $\tau/2$ during one interval.
From its endpoint, at time $(j-1)\tau/2$,  start the unused agent with
the smallest index, say $i_1(j)$; this agent can cover a distance of
$\frac{\tau}{2} \frac{1}{i_1(j)}$ during the interval. Continue in the same
way using new agents, all starting at time $(j-1)\tau/2$, until the
entire circle $C$ is covered; let the index of the last agent used be $i_2(j)$.
The covering condition can be written as:
\begin{equation} \label{E1}
\frac{\tau}{2} \left(1 + \sum_{i=i_1(j)}^{i_2(j)} \frac{1}{i} \right)
\geq 1, \textup{ or equivalently, }
1 + \sum_{i=i_1(j)}^{i_2(j)} \frac{1}{i} \geq \frac{2}{\tau}.
\end{equation}
For example, if $\tau=2/3$: $j=1$ requires agents $a_1$ through
$a_{11}$, since $H_{11} \geq 3$, but  $H_{10} <3$;
$j=2$ requires agents $a_1$ and agents $a_{12}$ through
$a_{85}$, since $1+(H_{85}-H_{11}) \geq 3$, but $1+(H_{84}-H_{11}) <3$.

We now bound from above the total number $k$ of distinct agents used
\ie, with speeds $1/i$, for $i=1,\ldots,k$.
Observe that the covering condition~\eqref{E1} may lead to
overshooting the target. Because the harmonic series has decreasing
terms, the overshooting error cannot exceed the term
$\frac{1}{i_2(1)+1}$ for $\tau=1$, namely $1/5$ (the overshooting for
$\tau=1$ is only $\frac13 -\frac14=\frac{1}{12} < \frac15$).
So inequality~\eqref{E1} becomes
\begin{equation} \label{E2}
\frac{2}{\tau} \leq 1 + \sum_{i=i_1(j)}^{i_2(j)} \frac{1}{i} \leq
\frac{2}{\tau} + \frac{1}{5}.
\end{equation}

Recall that $t=m\tau$.
By adding inequality~\eqref{E2} over all $2m$ time intervals yields
(in equivalent forms)
\begin{equation} \label{E3}
H_k -1 + \frac{8m}{5} \leq \frac{4m}{\tau}, \textup{  or  }
H_k \leq \frac{4t}{\tau^2} + 1 - \frac{8t}{5\tau}.
\end{equation}

For $t \geq \tau$ we have $1 \leq \frac{8t}{5\tau}$.
Since $\ln{k} \leq H_k$, it follows from~\eqref{E3} that
$$ \ln{k} \leq \frac{4t}{\tau^2}, \textup{ or } k \leq e^{4t/\tau^2},
$$
as required.
\end{proof}

\section{Bidirectional Circle Patrolling}
\label{sec:2/3}

\paragraph{A new schedule for closed fence patrolling.}
Czyzowicz~\etal~\cite[Theorem 5]{CGKK11} showed that for $k=3$
there exist maximum speeds $v_1,v_2,v_3$ and a schedule that achieves
a shorter idle time than both algorithm $\A_1$ and $\A_2$,
namely $35/36$ versus $12/11$ and $1$. We extend this result for all $k \geq 4$.

We propose a new algorithm, $\A_3$ for maximum speeds
$v_1 \geq v_2 \geq \ldots \geq v_k>0$, and then show that
$\A_3$ outperforms both $\A_1$ and $\A_2$ for
some speed settings for all $k\geq 4$.

We will need $v_1>v_2$ in this algorithm. Place the $k-1$ agents
$a_2,\ldots,a_k$ at equal distances, $x$ on the unit circle, and let
them move all clockwise perpetually at the same speed $v_k$; we say that
$a_2,\ldots,a_k$ make a ``train''. Let $a_1$ move back and forth
(\ie, clockwise and counterclockwise) perpetually on the moving arc of
length $1-(k-2)x$, \ie, between the start and the end of the train.
Refer to Fig.~\ref{fig:train}.
\begin{figure}[htbp]
    \begin{center}
	  \includegraphics[scale=0.55]{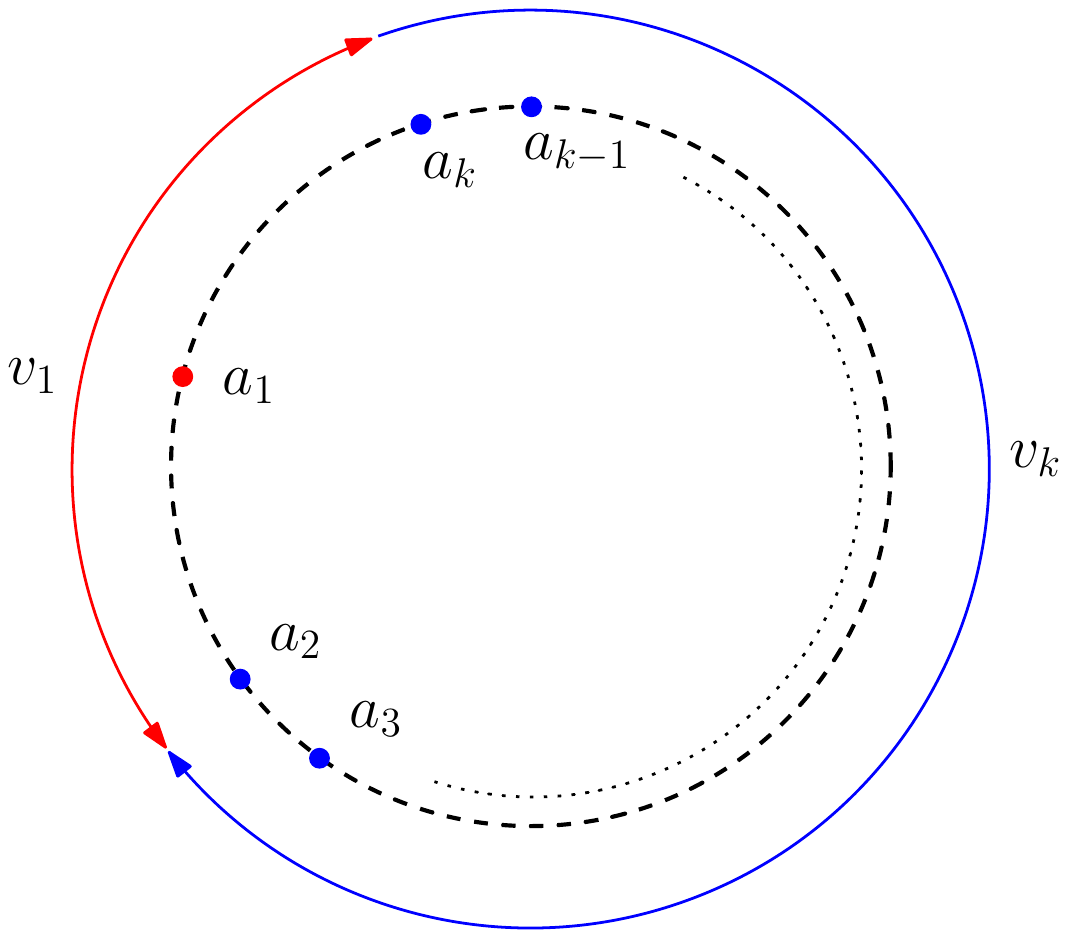}
    \end{center}
\caption{Train algorithm: the train $a_2,\ldots,a_k$ moving
  unidirectionally with speed $v_k$ and the bidirectional agent $a_1$
  with speed $v_1$.}
\label{fig:train}
\end{figure}

\begin{proposition} \label{prop:2/3}
For every $k \geq 4$, there exist maximum speeds $v_1 > v_2 \geq \ldots \geq v_k$
such that algorithm $\A_3$ achieves a shorter idle time than $\A_1$ and $\A_2$.
In particular, for large $k$, the idle time achieved by the train algorithm
is  about $2/3$ of those achieved by $\A_1$ and $\A_2$.
\end{proposition}
\begin{proof}
Consider the speed setting $v_1=a$, $v_2=\ldots=v_k=b$, where $a>b>0$,
and $\max_{1 \leq i \leq k} i v_i =kb$ (\ie, $a \leq kb$).
Put $y=1 -(k-2)x$.
To determine the idle time, $x/b$, write:
$$ [1 -(k-2)x] \left(\frac{1}{a-b} + \frac{1}{a+b} \right) =
\frac{x}{b}, \textup{ or equivalently,  }
\frac{2ay}{a^2-b^2} = \frac{1-y}{(k-2)b}. $$
Solving for $x/b$ yields
$$ \idle(\A_3)=\frac{2a}{a^2-b^2 + 2(k-2)ab}. $$
For our speed setting, we also have
$$ \idle(\A_1)= \frac{2}{a + (k-1)b}, \textup{ and }
\idle(\A_2)= \frac{1}{kb}. $$
Write $t=a/b$. It can be checked that for $k \geq 4$,
$ \idle(\A_3) \leq \idle(\A_1)$ and $ \idle(\A_3) \leq \idle(\A_2)$
when $a^2-b^2-4ab \geq 0$, \ie, $t \geq 2+\sqrt5$.
In particular, for $a=1$, and $b=1/k$ (note that $a \leq kb$), we have
$$ \idle(\A_3) = \frac{2}{1-1/k^2 + 2(k-2)/k }
\underset{k \to \infty}{\longrightarrow} \, \frac{2}{3}, $$
while
\begin{equation*}
\idle(\A_1) = \frac{2}{1 + (k-1)/k}
\underset{k \to \infty}{\longrightarrow} \, 1, \textup{ and }
\idle(\A_2) = \frac{1}{k(1/k)} =1.
\tag*{\qedhere}
\end{equation*}
\end{proof}

\paragraph{Useless agents for circle patrolling.}
Czyzowicz~\etal~\cite{CGKK11} showed that for $k=2$ there are
maximum speeds for which an optimal schedule does not use one
of the agents. Here we extend this result for all $k \geq 2$:

\begin{proposition} \label{prop:useless}
{\rm (i)} For every $k \geq 2$, there exist maximum speeds
$v_1 \geq v_2 \geq \ldots \geq v_k>0$ and an optimal schedule
for the circle with these speeds that does not use
up to $k-1$ of the agents $a_2,\ldots,a_k$.

\noindent {\rm (ii)} In contrast, for a segment, any optimal schedule
must use all agents.
\end{proposition}
\begin{proof}
(i) Let $v_1=1$ and $v_2=\ldots=v_k=\eps/k$, for a small positive $\eps \leq 1/300$,
and $C$ be a unit circle. Obviously by using agent $a_1$ alone
(moving perpetually clockwise) we can achieve unit idle time. Assume
for contradiction that there exists a schedule achieving an idle
time less than $1$.
Let $a_1(t)=t \mod 1$ denote the position of agent $a_1$ at time $t$.
Assume without loss of generality that $a_1(0)=0$ and consider the
time interval $[0,2]$. For $2 \leq i \leq k$, let $J_i$ be the
interval of points visited by agent $a_i$ during the time interval
$[0,2]$, and put $J=\cup_{i=2}^k J_i$.
We have $|J_i| \leq 2\eps/k$, thus $|J| \leq 2\eps$.
We make the following observations:

\begin{enumerate}
\item
$a_1(1) \in [-2\eps, 2\eps]$.
Indeed, if $a_1(1) \notin [-2\eps, 2\eps]$, then either some point in
$[-2\eps, 2\eps]$ is not visited by any agent during the time interval
$[0,1]$, or some point in $C \setminus [-2\eps, 2\eps]$ is not visited
by any agent during the time interval $[0,1]$.

\item
$a_1$ has done almost a complete (say, clockwise) rotation along $C$
  during the time interval $[0,1]$, \ie, it starts at
$0 \in [-2\eps, 2\eps]$ and ends in $[-2\eps, 2\eps]$, otherwise some
  point in $C \setminus [-2\eps, 2\eps]$ is not visited during the
  time interval $[0,1]$.

\item
$a_1(2) \in [-4\eps, 4\eps]$, by a similar argument.

\item
$a_1$ has done almost a complete rotation along $C$ during the
time interval $[1,2]$, \ie, it starts in $[-2\eps, 2\eps]$ and
ends in $[-4\eps, 4\eps]$. Moreover this rotation must be in the same
 clockwise sense as the previous one, since otherwise there would
 exist points not visited for at least one unit of time.
\end{enumerate}

Pick three points $x_1,x_2,x_3 \in C \setminus J$
close to $1/4$, $2/4$, and $3/4$,
respectively, \ie, $|x_i - i/4| \leq 1/100$, for $i=1,2,3$.
By Observations 2 and 4, these three points must be visited by $a_1$ in
the first two rotations during the time interval $[0,2]$ in the order
$x_1,x_2,x_3,x_1,x_2,x_3$.
Since $a_1$ has unit speed, successive visits to $x_1$
are separated in time by at least one time unit, contradicting
the assumption that the idle time of the schedule is less than $1$.

\smallskip
(ii) Given $v_1 \geq v_2 \geq \ldots \geq v_k > 0$,
assume for contradiction that there is an optimal
guarding schedule with unit idle time for a segment $s$ of maximum
length that does not use agent $a_j$ (with maximum speed $v_j$), for
some $1 \leq j \leq k$. Extend $s$ at one end by a subsegment of length
$v_j/2$ and assign $a_j$ to this subsegment to move back and
forth from one end to the other, perpetually. We now have a guarding
schedule with unit idle time for a segment longer than $s$, which is a
contradiction.
\end{proof}

\paragraph{Overtaking other agents.}
Consider an optimal schedule for circle patrolling (with unit idle time)
for the agents in the proof of Proposition~\ref{prop:useless},
with $v_1=1$ and $v_2=\ldots=v_k=\eps/k$, in which all agents move clockwise
at their maximum speeds. Obviously $a_1$ will overtake all other agents during
the time interval $[0,2]$. Thus there exist settings in which if all
$k$ agents are used by a patrolling algorithm, then some agent(s) need
to overtake (pass) other agent(s).
Observe however that overtaking can be easily avoided in this setting
by not making use of any of the agents $a_2,\ldots,a_k$.

\section{An Improved Idle Time for Open Fence Patrolling} \label{sec:24/25}

Kawamura and Kobayashi~\cite{KK12} showed that algorithm $\A_1$ 
by Czyzowicz~\etal~\cite{CGKK11} does not always produce an optimal
schedule for open fence patrolling. They presented two counterexamples: 
their first example uses $6$ agents and achieves an idle time of
$\frac{41}{42} \, \idle(\A_1)$; 
their second example uses $9$ agents and achieves an idle time of
$\frac{99}{100} \, \idle(\A_1)$.  
By replicating the strategy from the second example with a number of
agents larger than $9$, \ie, iteratively using blocks of agents, we
improve the ratio to $24/25+\eps$ for any $\eps>0$.   
We need two technical lemmas to verify this claim. 
\begin{figure}[htbp]
    \begin{center}
          \includegraphics[scale=0.55]{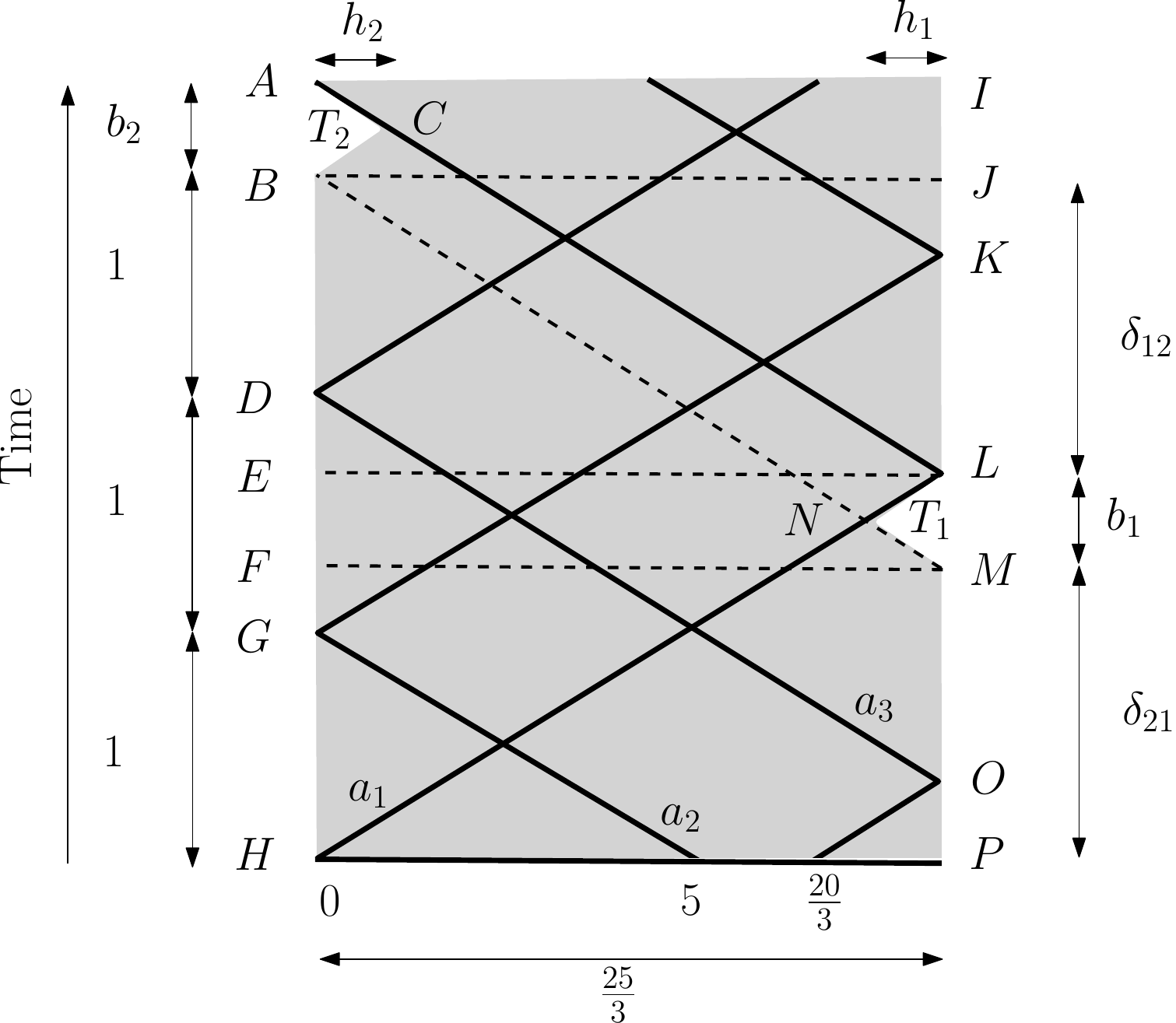}
    \end{center}
\caption{Three agents each with a speed of $5$ patrolling a fence of length
$25/3$; their start positions are $0$, $5$, and $20/3$, respectively.
Figure is \emph{not} to scale.}
\label{fig:LemmaDiagram1}
\end{figure}

\begin{lemma} \label{lem:1} Consider a segment of length
$L=\frac{25}{3}$ such that three agents $a_1,a_2,a_3$ are patrolling
perpetually each with speed of $5$ and generating an alternating sequence of
uncovered triangles $T_2,T_1,T_2,T_1,\ldots$, as shown in the
position-time diagram in Fig.~\ref{fig:LemmaDiagram1}. Denote the
vertical distances between consecutive occurrences of $T_1$ and $T_2$
by $\delta_{12}$ and between consecutive occurrences of $T_2$ and $T_1$
by $\delta_{21}$. Denote the bases of $T_1$ and $T_2$ by $b_1$ and $b_2$ respectively,
and the heights of $T_1$ and $T_2$ by $h_1$ and $h_2$ respectively. Then
\begin{enumerate} \itemsep 0pt
 \item [{\rm (i)}] $\frac{10}{3}$ is a period of the schedule.
 \item [{\rm (ii)}] $T_1$ and $T_2$ are congruent; further, $b_1 = b_2 = \frac{1}{3}$,
$\delta_{12}=\delta_{21} = \frac{4}{3}$, and $h_1 = h_2 = \frac{5}{6}$.
\end{enumerate}
 \end{lemma}

\begin{proof} (i) Observe that $a_1$, $a_2$ and $a_3$ reach the left
endpoint of the segment at times $2 (25/3) /5=10/3$, $5/5=1$, and
$(25/3 + 5/3)/5=2$, respectively. During the time interval $[0,10/3]$,
each agent traverses the distance $2L$ and the positions
and directions of the agents at time $t=10/3$ are the same as those at
time $t=0$. Hence $10/3$ is a period for their schedule.

(ii) Since $AL \parallel BM$ and $AB \parallel LM$, we have $b_1 = b_2$.
Since $L$ is the midpoint of $IP$, we have
$\delta_{12}+b_2=\delta_{21}+b_1$, thus $\delta_{12}=\delta_{21}$.
Since all the agents have same speed, $5$, all the
trajectory line segments in the position-time diagram have the same
slope, $1/5$. Hence $\angle{BAC} = \angle{ABC} = \angle{MLN} = \angle{LMN}$.
Thus, $T_1$ is similar to $T_2$. Since $b_1 = b_2$, $T_1$ is congruent
to $T_2$, and consequently $h_1 = h_2$.

Put $b=b_1$, $h=h_1$, and  $\delta=\delta_{12}$.
Recall from (i) that $|AH|=10/3$. By construction, we have $|BD|=1$, thus
$|BH|=|BD|+|DG|+|GH|=1+1+1=3$. We also have $|AH| = b+ |BH|$, thus
$b=10/3-3=1/3$. Since $L$ is the midpoint of $IP$,
we have $\delta+b=5/3$, thus $\delta=5/3 -b=4/3$.

Let $x(N)$ denote the $x$-coordinate of point $N$; then  $x(N)+h = 25/3$.
To compute $x(N)$ we compute the intersection of the two segments $HL$
and $BM$. We have $H=(0,0)$, $L=(25/3,5/3)$, $B=(0,3)$, and $M=(25/3,4/3)$.
The equations of $HL$ and $BM$ are $HL$~:~$x=5y$ and $BM$~:~$x+5y=15$,
and solving for $x$ yields  $x = 15/2$, and consequently
$h= 25/3-15/2=5/6$.
\end{proof}

\begin{figure}
\centering
\begin{subfigure}{.5\textwidth}
  \centering
  \includegraphics[scale=0.45]{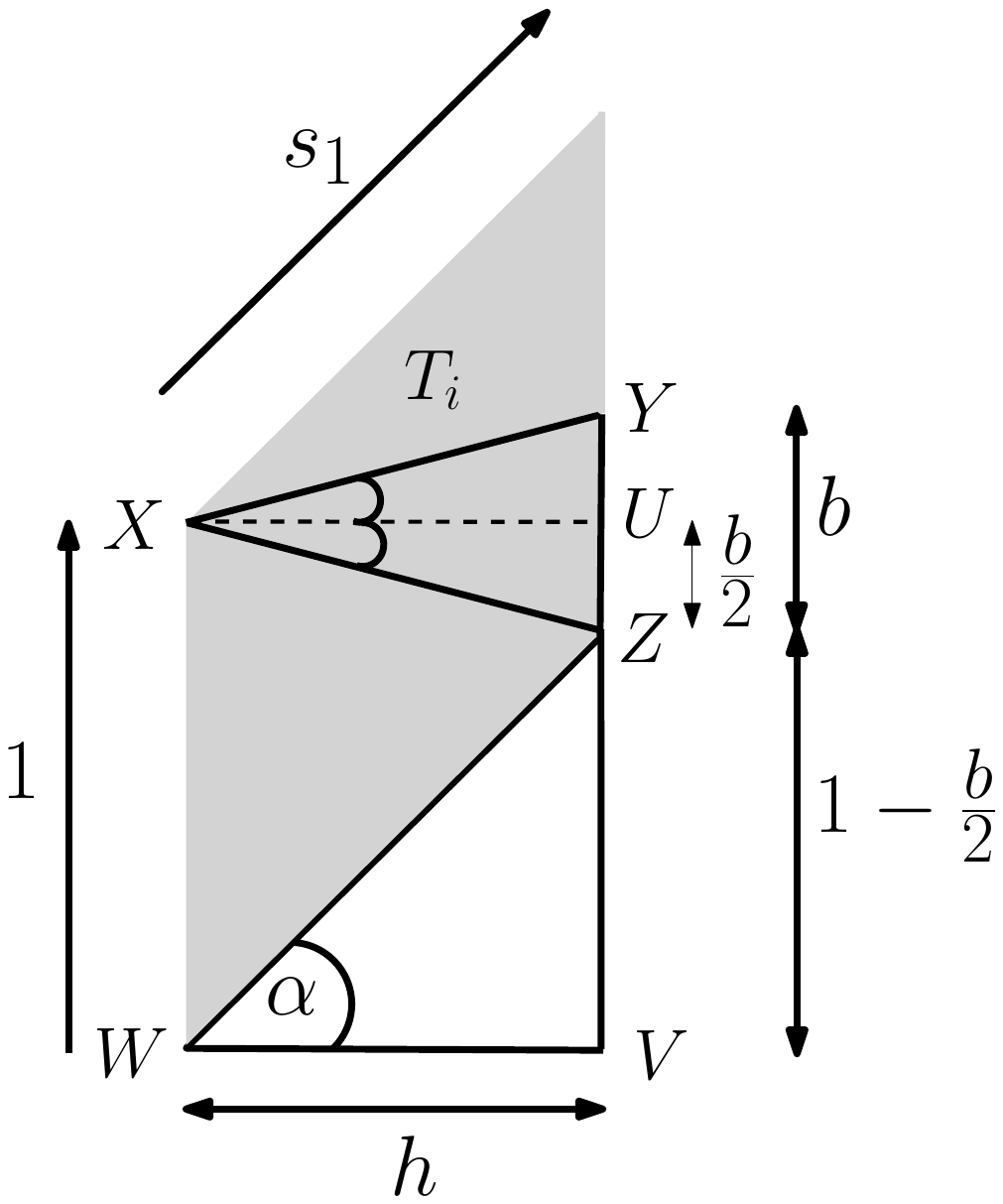}
   \label{fig:sub1}
\end{subfigure}%
\begin{subfigure}{.5\textwidth}
  \centering
  \includegraphics[scale=0.45]{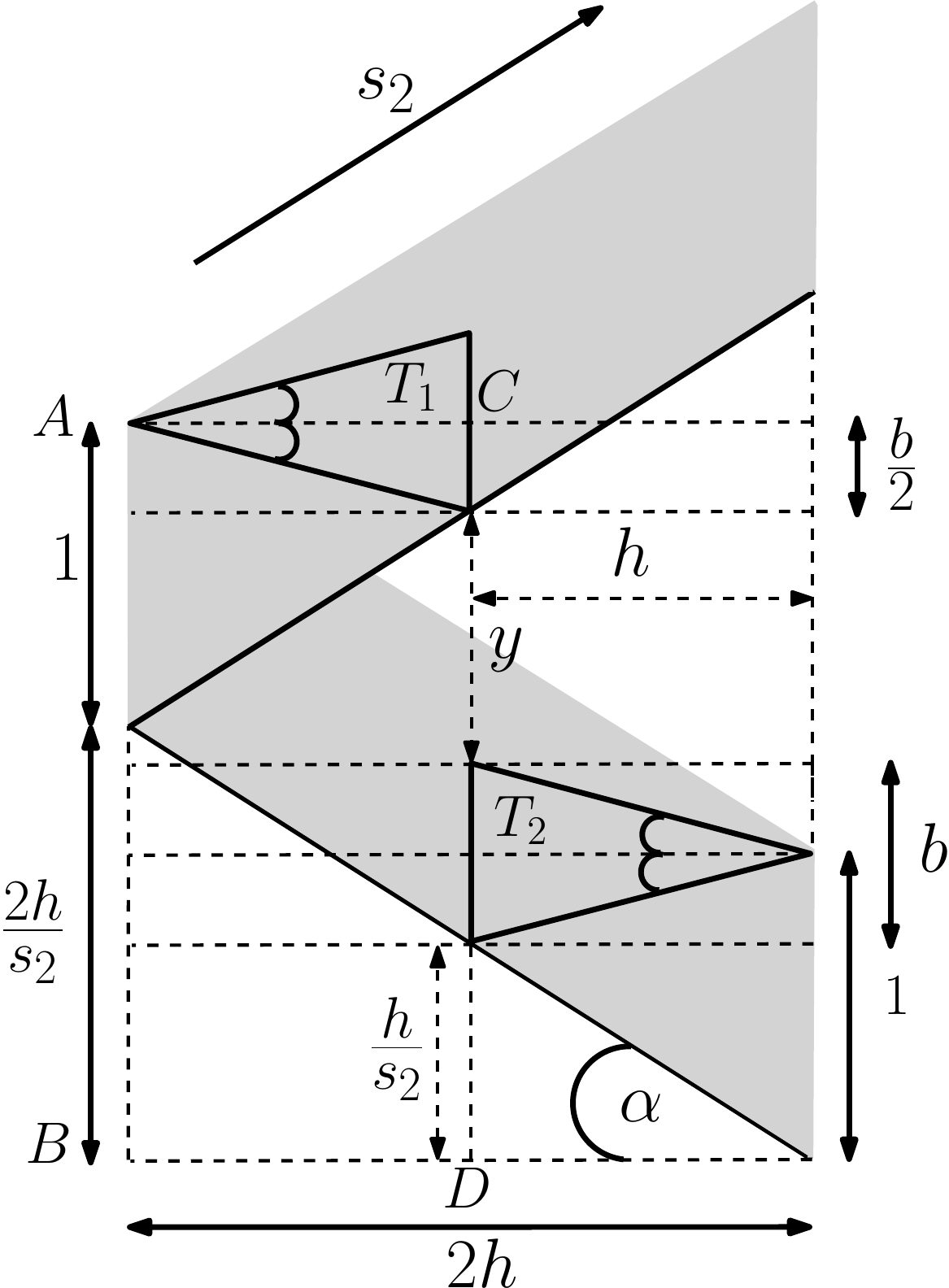}
   \label{fig:sub2}
\end{subfigure}
\caption{Left: agent covering an uncovered triangle $T_i$.
Right: agent covering an alternate sequence of congruent triangles
$T_1,T_2$, with collinear bases.}
\label{fig:LemmaDiagram2}
\end{figure}

\begin{lemma} \label{lem:2}
{\rm (i)} Let $s_1$ be the speed of an agent needed to cover an
uncovered isosceles triangle $T_i$; refer to Fig.~\ref{fig:LemmaDiagram2}~(left).
Then $s_1 = \frac{h}{1-b/2}$,
where $b<1$ and $h$ are the base and height of $T_i$, respectively.

{\rm (ii)} Let $s_2$ be the speed of an agent needed to cover an alternate
sequence of congruent isosceles triangles $T_1,T_2$ with bases on same vertical line;
refer to Fig.~\ref{fig:LemmaDiagram2}~(right).
Then $s_2 = \frac{h}{3b/2+y-1} $ where $y$ is the vertical
distance between the triangles, $b<1$ is the base and $h$ is the height
of the congruent triangles.
\end{lemma}

\begin{proof}
(i) In Fig.~\ref{fig:LemmaDiagram2}~(left), $\tan \alpha = 1/s_1$,
$|UZ| = b/2$, hence $|VZ| = 1-b/2$. Also, $\frac{|VZ|}{|WV|} =
\tan \alpha = \frac{1-b/2}{h} = \frac{1}{s_1}$, which yields $s_1 = \frac{h}{1-b/2}$.

(ii) In Fig.~\ref{fig:LemmaDiagram2}~(right), $|AB| = 1 + \frac{2h}{s_2}$.
Also, $|CD| = \frac{b}{2} + y + b + \frac{h}{s_2}$.
Equating $1 + \frac{2h}{s_2} = \frac{3b}{2} + y + \frac{h}{s_2}$ and
solving for $s_2$, we get $s_2 = \frac{h}{3b/2+y-1}$.
\end{proof}

\begin{figure}[htbp]
    \begin{center}
          \includegraphics[scale=0.6]{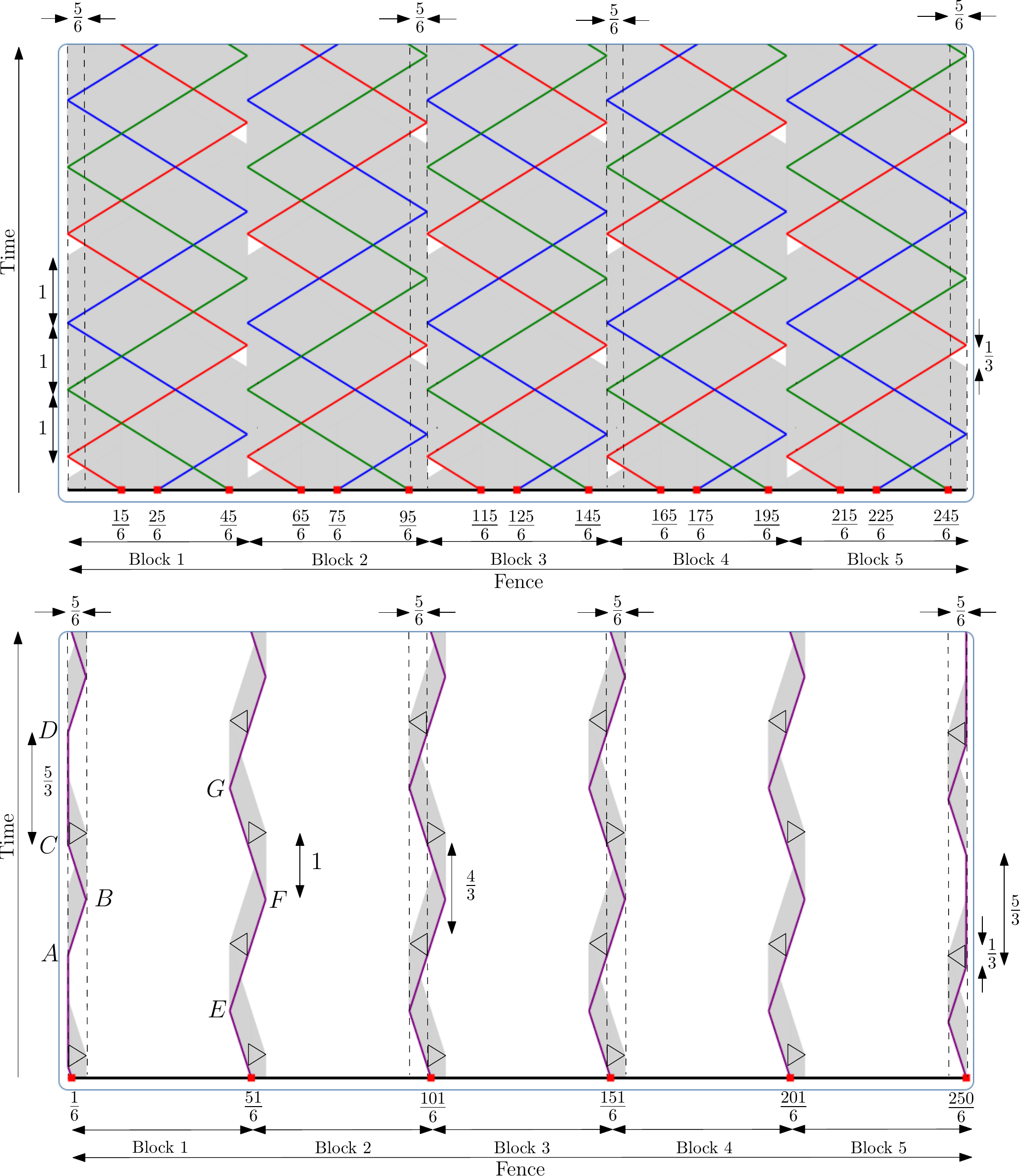}
          \includegraphics[scale=0.6]{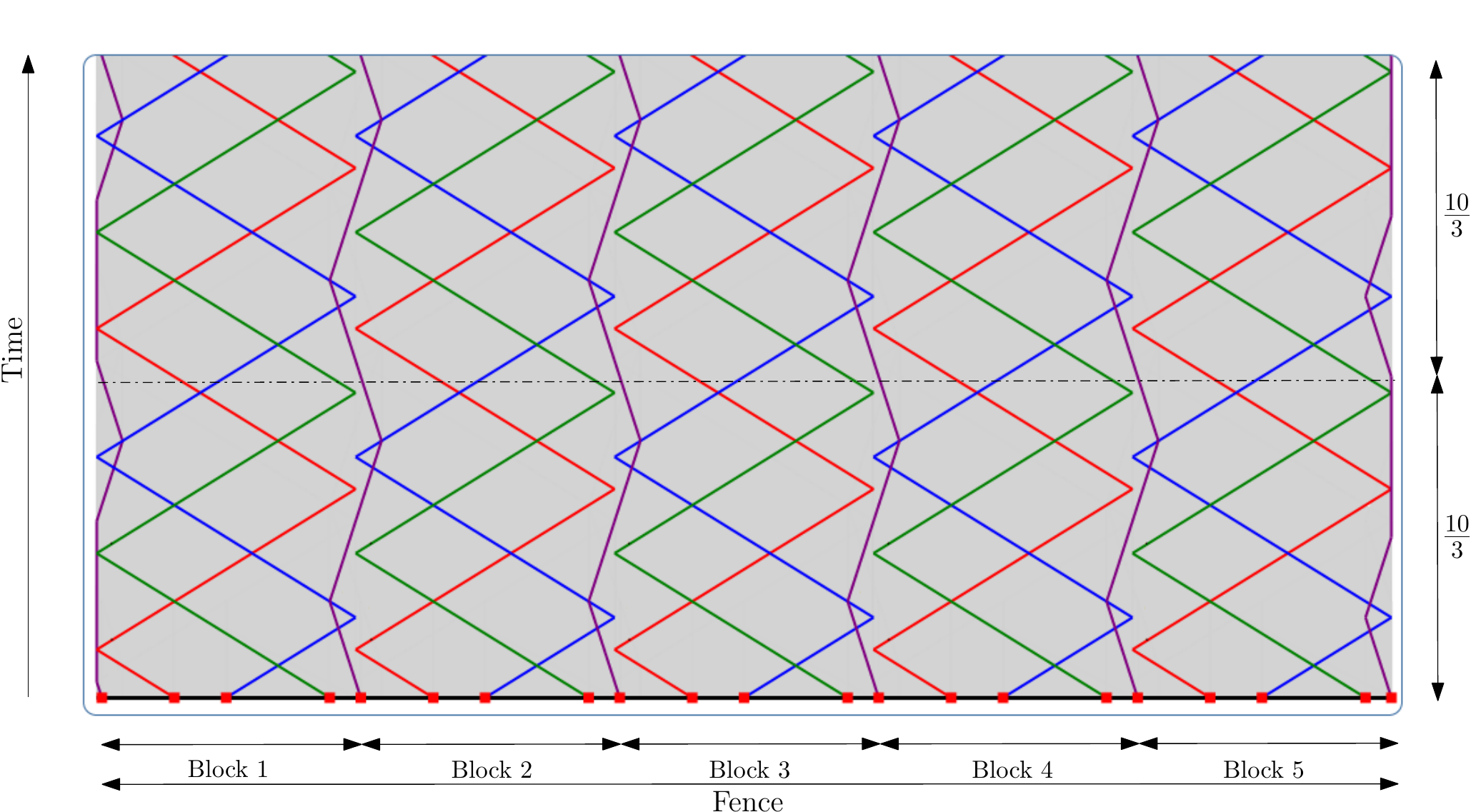}
    \end{center}
\caption{Top: iterative construction with $5$ blocks; each block has
three agents with speed $5$.
Middle: six agents with speed $1$.
Bottom: patrolling strategy for $5$ blocks using $21$ agents for two
time periods (starting at $t=1/3$ relative to
Fig.~\ref{fig:LemmaDiagram1}); the block length is $25/3$ and the
time period is $10/3$.}
\label{fig:24/25}
\end{figure}

\begin{theorem} \label{thm:24/25}
For every integer $x \geq 2$, there exist $k=4x+1$ agents with
$\sum_{i=1}^{k} v_i=16x+1$ and a guarding schedule for a segment of length $25x/3$.
Alternatively, for every integer $x \geq 2$ there exist $k=4x+1$
agents with suitable speeds $v_1,\ldots,v_k$,
and a guarding schedule for a unit segment that achieves idle
time at most $\frac{48x+3}{50x} \, \frac{2}{\sum_{i=1}^{k} v_i}$.
In particular, for every $\eps>0$, there exist $k$ agents
with suitable speeds $v_1,\ldots,v_k$,
and a guarding schedule for a unit segment that achieves idle
time at most
$\left(\frac{24}{25} +\eps\right) \, \frac{2}{\sum_{i=1}^{k} v_i}$.
\end{theorem}
\begin{proof} Refer to Fig.~\ref{fig:24/25}.
We use a long fence divided into $x$ blocks; each block is of length $25/3$.
Each block has 3 agents each of speed 5 running in zig-zag fashion.
Consecutive blocks share one agent of speed $1$ which covers the
uncovered triangles from the trajectories of the zig-zag agents in the
position-time diagram. The first and the last block use two agents of
speed $1$ not shared by any other block. The setting of these speeds
is explained below.

From Lemma~\ref{lem:1}(ii), we conclude that all the uncovered triangles generated
by the agents of speed 5 are congruent and their base is $b = 1/3$ and
their height is $h =5/6$. By Lemma~\ref{lem:2}(i), we can set the speeds of
the agents not shared by consecutive blocks to $s_1 = \frac{5/6}{1-1/6} = 1$.
Also, in our strategy, Lemma~\ref{lem:1}(ii) yields $y = \delta = 4/3$.
Hence, by Lemma~\ref{lem:2}(ii), we can set the speeds of the agents
shared by consecutive blocks to $s_2 = \frac{5/6}{1/2 + 4/3 -1} = 1$.

In our strategy, we have 3 types of agents:
agents running with speed~5 as in Fig.~\ref{fig:24/25}~(top),
unit speed agents not shared by 2 consecutive blocks and
unit speed agents shared by two consecutive blocks as in
Fig.~\ref{fig:24/25}~(middle).
By Lemma~\ref{lem:1}(i), the agents of first type have period $10/3$.
In Fig.~\ref{fig:24/25}~(middle), there
are two agents of second type and both have a similar trajectory.
Thus, it is enough to verify for the leftmost unit speed agent.
It takes $5/6$ time from $A$ to $B$ and again $5/6$ time from
$B$ to $C$. Next, it waits for $5/3$ time at $C$. Hence after
$5/6+5/6+5/3 = 10/3$ time, its position and direction at $D$ is same as
that at $A$. Hence, its time period is $10/3$.
For the agents of third type, refer to Fig.~\ref{fig:24/25}~(middle): it
takes $10/6$ time from $E$ to $F$
and $10/6$ time from $F$ to $G$. Thus, arguing as above, its time period is
$10/3$. Hence, overall, the time period of the strategy is $10/3$.

For $x$ blocks, we use $3x+(x+1) = 4x+1$ agents.
The sum of all speeds is $5(3x) + 1(x+1) = 16x+1$ and the total fence
length is $\frac{25x}{3}$. The resulting ratio is
$\rho = \frac{16x+1}{2} / \frac{25x}{3} = \frac{48x+3}{50x}$. For
example, when $x=2$ we reobtain the bound of Kawamura and
Kobayashi~\cite{KK12} (from their 2nd example), when $x=39, \rho =
\frac{100}{104}$ and further on,
$\rho \underset{x \to \infty}{\longrightarrow} \, \frac{24}{25}$.
Thus an idle time of at most
$\left(\frac{24}{25} +\eps\right) \, \frac{2}{\sum_{i=1}^{k} v_i}$
can be achieved for every given $\eps>0$, as required.
\end{proof}

\paragraph{Acknowledgements.}
We sincerely thank Akitoshi Kawamura for
generously sharing some technical details concerning their patrolling algorithms.
We also express our satisfaction with the software package
\emph{JSXGraph, Dynamic Mathematics with JavaScript}, used in our experiments.

\end{document}